\documentclass[aps,twocolumn,superscriptaddress,longbibliography,nofootinbib]{revtex4-2}

\usepackage{amsmath,amssymb}
\usepackage{amsthm}
\usepackage{graphicx}
\usepackage{bm}
\usepackage[colorlinks=true,linkcolor=blue,citecolor=blue,urlcolor=blue]{hyperref}
\usepackage{orcidlink}

\providecommand{\ket}[1]{|#1\rangle}
\providecommand{\bra}[1]{\langle#1|}

\newtheorem{theorem}{Theorem}
\newtheorem{proposition}{Proposition}

\newtheorem{corollary}{Corollary}
\theoremstyle{remark}

\newcommand{\Ncal}{\mathcal{N}}

\begin{document}

\title{Phase information beyond entanglement sudden death in
coherence-to-entanglement conversion under post-gate noise}

\author{Asad Ali\orcidlink{0000-0001-9243-417X}}
\email{asal68826@hbku.edu.qa}
\affiliation{Qatar Center for Quantum Computing, College of Science and Engineering, Hamad Bin Khalifa University, Doha, Qatar}
\author{Hashir~Kuniyil\orcidlink{0000-0003-0338-1278}}
\affiliation{Qatar Center for Quantum Computing, College of Science and Engineering, Hamad Bin Khalifa University, Doha, Qatar}
\author{M.T Rahim\orcidlink{0000-0003-1529-928X}}
\affiliation{Qatar Center for Quantum Computing, College of Science and Engineering, Hamad Bin Khalifa University, Doha, Qatar}
\author{Saif Al-Kuwari\orcidlink{0000-0002-4402-7710}}
\affiliation{Qatar Center for Quantum Computing, College of Science and Engineering, Hamad Bin Khalifa University, Doha, Qatar}

\date{\today}

\begin{abstract}
An ideal CNOT maps the phase of a coherent qubit onto the coherence between
$\ket{00}$ and $\ket{11}$ of a two-qubit state, producing an output that
carries both entanglement and estimable phase information. We ask how
post-gate noise degrades these two quantities, and find that they are not
lost together. For the phase-encoded X states generated by the protocol, the
negativity is a thresholded difference of the surviving coherence $z=f\kappa$
and a population penalty $g$, vanishing once $f\kappa\le g$, while the phase
quantum Fisher information (QFI) is the smooth ratio $F_\phi=4z^2/(a+b)$,
which stays positive for any nonzero coherence. As a result there is an exact
region of state space in which the output is separable but still
phase-sensitive. We characterize this region, give the residual QFI
$F_\phi^\star=4g_\star^2/(1-2g_\star)$ at entanglement death, and show that
channels reaching death at the same coordinate share this residual, with
global and independent local depolarization forming one such class and
$F_\phi^\star=1/6$ at maximal input coherence. Four standard channels appear
as trajectories through this common geometry, and asymmetric population
transfer adds a third coordinate that changes the entanglement but leaves the
QFI unchanged, which marks where the two-coordinate description applies. We
identify a measurement that attains the bound and compare with a direct
single-qubit probe, which is more precise under matched exposure; the results
are therefore reference benchmarks for phase-information retention, not a
claim of metrological advantage.
\end{abstract}

\maketitle

\section{Introduction}
\label{sec:intro}

A single-qubit phase, written as the phase of an off-diagonal coherence
$\cos\theta\ket0+e^{i\phi}\sin\theta\ket1$, can be mapped by a CNOT gate onto
the coherence between $\ket{00}$ and $\ket{11}$ of a two-qubit state,
converting a local superposition into a bipartite resource
\cite{Baumgratz2014,StreltsovRMP,Streltsov2015,Ma2016,Killoran2016,Zhu2018}.
The output then carries two operationally distinct quantities: the
entanglement generated by the conversion, and the phase itself, whose
estimation precision is set by the quantum Fisher information (QFI) through
the quantum Cram\'er--Rao bound \cite{Helstrom,Holevo,BraunsteinCaves,
Paris2009,TothApellaniz,LiuReview}. We take the CNOT to be ideal and the
noise to act after the gate, and ask whether such noise degrades the two
quantities together, that is, whether the loss of entanglement coincides
with the loss of phase information.

It does not. The two quantities are functionals of the same noisy state but
enter through different operations, and those operations reach zero at
different points. This is the result of the paper, developed below in exact
form.

Four established lines of research frame the question. The conversion of coherence
into entanglement by incoherent operations is foundational, including the
quantitative correspondence between input coherence and output entanglement
\cite{Streltsov2015,Ma2016,Killoran2016,Zhu2018}. The QFI is tied to
coherence and asymmetry, and single-qubit phase estimation under noise is
standard \cite{Paris2009,Zhong2013,ChapeauBlondeau2015}. General methods for
the QFI of finite-dimensional mixed states are well established
\cite{Paris2009,LiuReview}, so evaluating the QFI of a $2\times2$ block is
not itself an advance. And entanglement decay under standard channels,
including entanglement sudden death (ESD), is thoroughly studied
\cite{Peres,VidalWerner,Companion}. Here we provide a
\emph{same-state} analytical comparison: the phase QFI and the entanglement
evaluated along identical channel trajectories of one conversion protocol,
with the relation between them stated exactly rather than plotted.

The conceptual content is a contrast between two functional forms. Write the
surviving phase-bearing coherence as $z=f\kappa$, where
$\kappa=\tfrac12\sin2\theta$ fixes the input and $f$ is the channel's
coherence-suppression factor, and write $g$ for the population the noise
deposits in the single-excitation subspace. Then, for the symmetric-leakage
family generated by the protocol,
\begin{equation}
\Ncal=\max\{0,\;z-g\},
\qquad
F_{\phi}=\frac{4z^{2}}{1-2g}.
\label{eq:headline}
\end{equation}
Negativity is a \emph{thresholded difference}: it vanishes at finite noise
whenever the population penalty $g$ catches up to the coherence $z$. The QFI
is a \emph{smooth ratio}: it is positive for any $z>0$ and merely
renormalized by the block weight $1-2g$. Between the point where negativity
vanishes and the point where coherence vanishes lies a region of state space
in which the output is \emph{separable but still phase-sensitive}. That
region is the main object of study.

The consequences we establish exactly are: the phase QFI is finite at and
beyond ESD (Sec.~\ref{sec:separation}); the residual sensitivity at the
death point is a function of the death coordinate alone, so
channels reaching death at the same coordinate share their residual QFI, and
the two depolarizing channels form such a class with $F_{\phi}^{\star}=
\tfrac16$ (Sec.~\ref{sec:death}); asymmetric population transfer adds a third
coordinate that changes the entanglement but not the QFI, marking the exact
boundary of the two-coordinate description (Sec.~\ref{sec:asym}); a specified
measurement attains the bound, and against a direct single-qubit probe the
converter offers no unconditional precision advantage, so the results are
information-retention benchmarks, not a metrological enhancement
(Sec.~\ref{sec:operational}); and a during-gate Lindblad model is bounded
from below by the post-gate idealization for the cases tested
(Sec.~\ref{sec:during}).

\emph{Relation to prior work of the authors.} A companion benchmark
\cite{Companion} determined how the generated \emph{entanglement} decays
under these channels, deriving the negativity coordinates $(f,g)$ and the
survival fingerprints. The present work asks a different, operational
question about the \emph{same} outputs: how much information about the
encoded phase remains, including after the state becomes separable. We cite
the negativity geometry rather than rederiving it, and the new content is the
phase QFI, the resource-separation region, the death-point sensitivity and
its equivalence classes, and the asymmetric-leakage boundary. All analytic
expressions were checked against direct density-matrix computation at machine
precision (Appendix~\ref{app:numerics}).

\section{Protocol and general phase-encoded state}
\label{sec:setup}

The protocol prepares qubit $A$ with coherence between $\ket0$ and $\ket1$,
qubit $B$ in $\ket0$, and applies an ideal CNOT. For any single-qubit input
encoding its unknown parameter as the phase of the off-diagonal element,
\begin{equation}
\rho_{A}(\phi)=
\begin{pmatrix} q & r\,e^{-i\phi}\\ r\,e^{i\phi} & 1-q\end{pmatrix},
\qquad 0\leq r\leq\sqrt{q(1-q)},
\label{eq:input}
\end{equation}
the CNOT with the incoherent ancilla produces the phase-encoded X state
\begin{equation}
\rho_{X}(\phi)=
\begin{pmatrix}
a & 0 & 0 & z\,e^{-i\phi}\\
0 & u & 0 & 0\\
0 & 0 & v & 0\\
z\,e^{i\phi} & 0 & 0 & b
\end{pmatrix},
\label{eq:generalX}
\end{equation}
initially with $a=q$, $b=1-q$, $u=v=0$, $z=r$, and after noise with general
$(a,b,u,v,z)$ obeying $a+b+u+v=1$ and $z^{2}\leq ab$. The phase enters only
the antidiagonal element, so the family is unitarily generated,
$\rho_{X}(\phi)=e^{-i\phi G}\rho_{X}(0)e^{i\phi G}$ with
$G=\ket{11}\bra{11}$, and the task is single-parameter estimation of $\phi$
with the fixed bounded generator $G$, whence $F_{\phi}\leq1$
\cite{TothApellaniz}. Numerical illustrations use the transmon-like scales
\cite{Krantz} $T_{1}=100~\mu$s, $T_{\varphi}=80~\mu$s, and CNOT error
$\varepsilon_{\rm gate}=0.5\%$ of the companion work \cite{Companion}, so
the figures are comparable across the series.

\section{Two resource functionals of the same state}
\label{sec:functionals}

The paper turns on the contrast between two functionals of the state
\eqref{eq:generalX}, so we state both generally before specializing.

\begin{proposition}[Block QFI]
\label{prop:block}
For the family \eqref{eq:generalX} with even-block weight $s\equiv a+b>0$,
\begin{equation}
F_{\phi}(\rho_{X})=\frac{4z^{2}}{a+b},
\label{eq:blockqfi}
\end{equation}
independently of the odd-sector populations $u$ and $v$.
\end{proposition}

The proof (Appendix~\ref{app:block}) is the standard spectral QFI evaluation
\cite{Paris2009,LiuReview}: $G$ annihilates $\ket{01},\ket{10}$, so only the
even-parity block contributes. The evaluation is routine qubit metrology; the
structural fact we use repeatedly is that $F_{\phi}$ depends on the state
\emph{only} through the pair $(z,s)$, and not through the imbalance $u-v$.

\begin{proposition}[Negativity]
\label{prop:neg}
For the same state, the partial transpose has a single candidate negative
eigenvalue, giving
\begin{equation}
\Ncal(\rho_{X})=\max\Bigl\{0,\;
\tfrac12\bigl[\sqrt{(u-v)^{2}+4z^{2}}-(u+v)\bigr]\Bigr\},
\label{eq:generalN}
\end{equation}
and the concurrence \cite{Wootters} satisfies $C=2\Ncal$ on this family.
\end{proposition}

This is the negativity of the companion benchmark \cite{Companion}, quoted
for contrast. The two functionals differ in one respect that underlies all
of the results below: \emph{the QFI is blind to the odd-sector split $u-v$,
whereas the negativity depends on it.} The QFI reads $(z,\,u+v)$ through a
smooth ratio; the negativity reads $(z,\,u,\,v)$ through a thresholded
difference.

\section{Symmetric-leakage geometry and resource separation}
\label{sec:separation}

The channels generated by the protocol share a structural feature: they
deposit \emph{equal} populations in the single-excitation states, $u=v\equiv
g$. We first record what this buys, then prove the central result.

\begin{theorem}[Common-coordinate condition]
\label{thm:coordinate}
If a channel maps the protocol output to \eqref{eq:generalX} with $u=v=g$,
then trace preservation fixes $s=a+b=1-2g$, and both resources are specified
by the pair $(z,g)$:
\begin{equation}
F_{\phi}=\frac{4z^{2}}{1-2g},
\qquad
\Ncal=\max\{0,\,z-g\},
\qquad
C=2\Ncal.
\label{eq:masters}
\end{equation}
\end{theorem}

\begin{proof}
$s=1-2g$ is $a+b+u+v=1$ with $u=v=g$. Setting $u=v=g$ in \eqref{eq:generalN}
collapses the square root to $2z$, giving $\Ncal=\max(0,z-g)$ and, via
$\sqrt{uv}=g$, $C=2\Ncal$. Equation \eqref{eq:blockqfi} with $s=1-2g$ gives
the QFI.
\end{proof}

The two resources are now functionals of the same two coordinates, read
through different operations: a smooth ratio for the QFI, a thresholded
difference for the negativity. Writing $z=f\kappa$ with $f$ the
coherence-suppression factor and $\kappa=\tfrac12\sin2\theta$ the input
coordinate, both resources are determined by the pair $(f\kappa,g)$. The
next theorem states the central consequence.

\begin{theorem}[Resource separation]
\label{thm:separation}
For symmetric leakage with $0\leq g<\tfrac12$, the output is separable
($\Ncal=0$) if and only if $f\kappa\leq g$, while it is phase-sensitive
($F_{\phi}>0$) if and only if $f\kappa>0$. Consequently the state space
partitions into three regions,
\begin{equation}
\underbrace{f\kappa>g}_{\text{entangled, sensitive}}\;\;
\underbrace{0<f\kappa\leq g}_{\text{separable, sensitive}}\;\;
\underbrace{f\kappa=0}_{\text{insensitive}},
\end{equation}
and the middle region, separable yet phase-sensitive, is nonempty for
every $g>0$.
\end{theorem}

\begin{proof}
Immediate from \eqref{eq:masters}: $\Ncal=\max(0,f\kappa-g)$ is zero exactly
when $f\kappa\leq g$, and $F_{\phi}=4(f\kappa)^{2}/(1-2g)$ is positive
exactly when $f\kappa>0$. For any $g>0$ the interval $0<f\kappa\leq g$ has
positive length, so the middle region is nonempty.
\end{proof}

Theorem~\ref{thm:separation} is the central result
[Fig.~\ref{fig:regions}]. It states an implication failure: $\Ncal=0$ does
\emph{not} imply $F_{\phi}=0$. Physically, entanglement requires the
phase-bearing coherence to exceed the population penalty it competes with in
the partial-transpose spectrum; phase sensitivity requires only that the
coherence be nonzero. Population noise can therefore push the output across
the separability boundary, extinguishing entanglement, without erasing the
coherence that carries $\phi$. In this family, entanglement death is not
phase-information death.

\begin{figure}[t]
\includegraphics[width=\columnwidth]{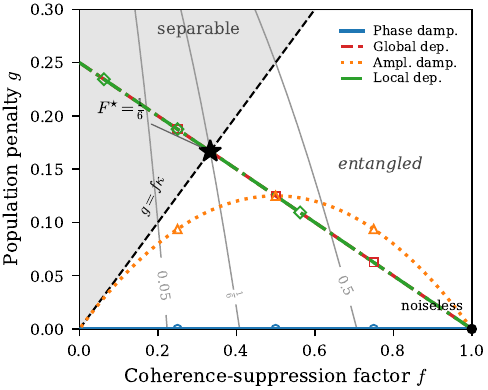}
\caption{Three-region geometry in the $(f,g)$ plane for a maximally
coherent input ($\theta=\pi/4$, $\kappa=\tfrac12$). Above the dashed line
$g=f\kappa$ the output is entangled and phase-sensitive; the shaded wedge
$0<f\kappa\leq g$ is \emph{separable but phase-sensitive}
(Theorem~\ref{thm:separation}); at $f\kappa=0$ the state is phase-insensitive.
Thin gray curves are iso-QFI contours $F_{\phi}=0.05$, $\tfrac16$, and $0.5$.
Colored curves are the four channel trajectories of Table~\ref{tab:fg},
starting at the noiseless point $(1,0)$, with open markers at noise values
$0.25$, $0.5$, and $0.75$. The star marks the common death point of global
and local depolarization on the boundary, with residual QFI
$F_{\phi}^{\star}=\tfrac16$ (Sec.~\ref{sec:death}); at this input, phase
damping and symmetric amplitude damping stay entangled until the coherence
itself vanishes.}
\label{fig:regions}
\end{figure}

\section{Channel trajectories}
\label{sec:channels}

The four post-gate channels of Ref.~\cite{Companion}, namely phase damping,
global depolarizing, amplitude damping, and independent local depolarizing
(Kraus sets in Appendix~\ref{app:channels}), all preserve the X form with
$u=v=g$, so each is a trajectory through the geometry of
Sec.~\ref{sec:separation}, entering via its pair $(f,g)$ (Table~\ref{tab:fg}).
We report the trajectories compactly; they are demonstrations of the common
geometry, not independent analyses, and the negativity coordinates are those
of the companion benchmark, now shown to govern the QFI as well.

\begin{table*}[t]
\caption{Channel coordinates: coherence-suppression factor $f$, population
penalty $g$, and block weight $s=1-2g$. The negativity reads $(f,g)$ through
the thresholded difference $f\kappa-g$; the phase QFI through the smooth
ratio $4f^{2}\kappa^{2}/(1-2g)$. Here $p$ ($\gamma$) is the channel noise
parameter and $\kappa=\tfrac12\sin2\theta$ fixes the input.}
\label{tab:fg}
\begin{ruledtabular}
\begin{tabular}{lcccc}
Channel & Parameter & $f$ & $g(\theta)$ & $s=1-2g$\\
\hline
Phase damping & $p$ & $1-p$ & $0$ & $1$\\
Global depolarizing & $p$ & $1-p$ & $p/4$ & $1-p/2$\\
Amplitude damping & $\gamma$ & $1-\gamma$ &
$\gamma(1-\gamma)\sin^{2}\theta$ & $1-2\gamma(1-\gamma)\sin^{2}\theta$\\
Local depolarizing & $p$ & $(1-p)^{2}$ & $p(2-p)/4$ & $1-p+p^{2}/2$\\
\end{tabular}
\end{ruledtabular}
\end{table*}

\begin{corollary}[Channel closed forms]
\label{cor:closed}
For the pure input with $\kappa=\tfrac12\sin2\theta$,
\begin{align}
F_{\rm phase}&=(1-p)^{2}\sin^{2}2\theta,\\
F_{\rm dep}&=\frac{(1-p)^{2}\sin^{2}2\theta}{1-p/2},\\
F_{\rm AD}&=\frac{(1-\gamma)^{2}\sin^{2}2\theta}
{1-2\gamma(1-\gamma)\sin^{2}\theta},\\
F_{\rm loc}&=\frac{(1-p)^{4}\sin^{2}2\theta}{1-p+p^{2}/2}.
\end{align}
\end{corollary}

Phase damping ($g=0$) is the zero-penalty baseline: coherence is suppressed
but no population penalty is created, so there is no finite-noise ESD and the
separable-but-sensitive region is never entered. Global depolarizing is the
cleanest case of the separation, in which finite ESD occurs while the QFI stays
positive. Local depolarizing moves through related coordinates at a different
rate, and shares its death-point QFI with global depolarizing
(Sec.~\ref{sec:death}). Amplitude damping tests the scope of symmetric
leakage: the symmetric two-sided model keeps $u=v$, and one-sided damping is
deferred to the three-coordinate treatment of Sec.~\ref{sec:asym}. Full
landscapes $F(\theta,x)$ with the ESD boundaries overlaid appear in
Fig.~\ref{fig:landscapes}; the QFI is continuous and strictly positive across
every boundary.

\begin{figure*}[t]
\includegraphics[width=\textwidth]{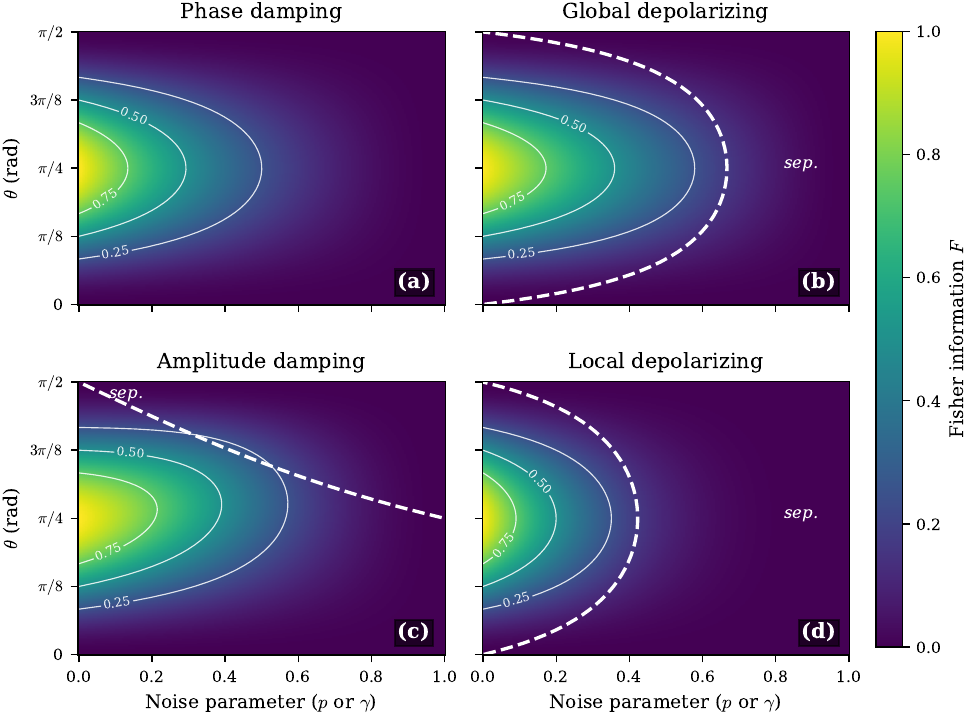}
\caption{QFI landscapes $F(\theta,\cdot)$ on a shared color scale with white
survival contours $F=0.25,0.5,0.75$: (a) phase damping, (b) global
depolarizing, (c) amplitude damping, (d) local depolarizing. White dashed
curves are the entanglement-death boundaries of the same states
\cite{Companion}, separable side marked ``sep.''; the QFI crosses every
boundary continuously and remains positive.}
\label{fig:landscapes}
\end{figure*}

\section{Death-point sensitivity and equivalence classes}
\label{sec:death}

The separation is sharpest at its boundary. At entanglement death the
symmetric coordinates satisfy $f_{\star}\kappa=g_{\star}$, and substituting
into \eqref{eq:masters} gives the residual sensitivity
\begin{equation}
F_{\phi}^{\star}=\frac{4g_{\star}^{2}}{1-2g_{\star}},
\label{eq:Fstar}
\end{equation}
a function of the death coordinate $g_{\star}$ alone. This turns the ESD
boundary, ordinarily just a separability threshold, into a metrological
benchmark and classifies channels within the symmetric-leakage family.

\begin{corollary}[Death-point equivalence]
\label{cor:equiv}
Because $h(g)=4g^{2}/(1-2g)$ is strictly increasing on $[0,\tfrac12)$, two
channels have equal residual QFI at entanglement death,
$F_{\phi,A}^{\star}=F_{\phi,B}^{\star}$, if and only if they share the death
coordinate, $g_{\star,A}=g_{\star,B}$.
\end{corollary}

\begin{corollary}[Depolarizing class]
\label{cor:onesixth}
Global and independent local depolarization both satisfy $g=(1-f)/4$
[from $p(2-p)=1-(1-p)^{2}$], hence reach death at the same coordinate
$g_{\star}=\kappa/(4\kappa+1)$ and share
\begin{equation}
F_{\phi}^{\star}=\frac{\sin^{2}2\theta}{(2\sin2\theta+1)(\sin2\theta+1)},
\end{equation}
equal to $\tfrac16$ at $\theta=\pi/4$. The two channels are not dynamically
identical, since they reach death at different noise strengths $p_{c}\neq
p_{c}^{\rm loc}$, but they carry the same residual sensitivity on the
boundary.
\end{corollary}

The value $\tfrac16$ is structural, not accidental: it is $h$ evaluated at
the death coordinate common to the depolarizing class at maximal input
coherence. Fig.~\ref{fig:death}(a) shows both channels falling on a single
residual curve as a function of $\theta$.

\begin{figure*}[t]
\includegraphics[width=\textwidth]{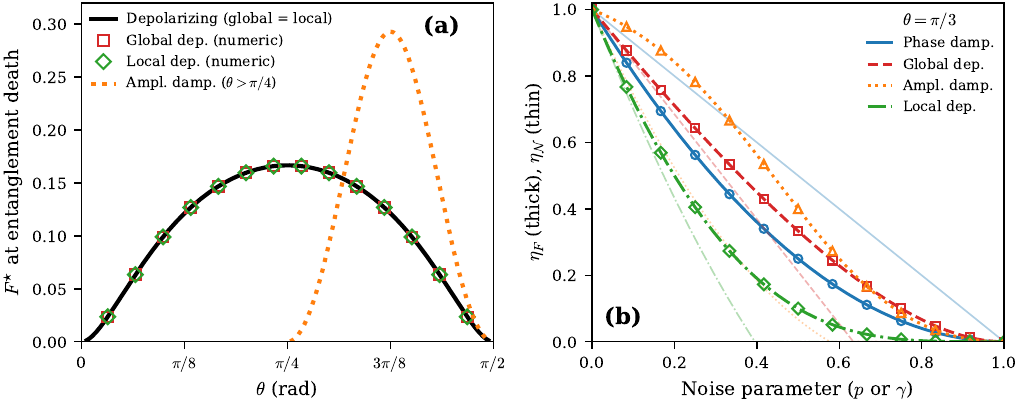}
\caption{(a)~Residual phase sensitivity $F_{\phi}^{\star}$ at entanglement
death versus input angle $\theta$. Global and independent local
depolarization (markers, numerically evaluated) lie on the common analytic
curve (solid), coinciding because they share the death coordinate
(Corollary~\ref{cor:equiv}); the value peaks at $F_{\phi}^{\star}=\tfrac16$
for the maximally coherent input $\theta=\pi/4$
(Corollary~\ref{cor:onesixth}). Amplitude damping (dotted) reaches death
only for $\theta>\pi/4$ and follows a different residual. (b)~QFI and
negativity survival fractions $\eta_{F},\eta_{\mathcal N}$ versus noise
parameter at $\theta=\pi/3$, showing that the QFI (thick) is retained past
the point where the negativity (thin) vanishes.}
\label{fig:death}
\end{figure*}

\section{Asymmetric leakage: the boundary of the geometry}
\label{sec:asym}

The two-coordinate description is exact for symmetric leakage and fails
otherwise, and the manner of its failure is itself a result. Write the
phase-block weight and the odd-sector imbalance
\begin{equation}
s=a+b=1-u-v,\qquad \delta=u-v.
\end{equation}
By Proposition~\ref{prop:block}, $F_{\phi}=4z^{2}/s$ depends only on $(z,s)$
and is \emph{invariant} under $\delta$. By Proposition~\ref{prop:neg},
$\Ncal=\tfrac12[\sqrt{\delta^{2}+4z^{2}}-(u+v)]$ is \emph{nondecreasing} in
$|\delta|$ at fixed $z$ and $u+v$.

\begin{proposition}[Symmetric leakage is extremal for entanglement]
\label{prop:extremal}
At fixed phase-bearing coherence $z$ and fixed total leakage $u+v$, the
negativity is minimized at $\delta=0$ and increases monotonically with
$|\delta|$, while the QFI is constant. Symmetric population transfer is
therefore the most destructive arrangement for entanglement at fixed
coherence and leakage.
\end{proposition}

\begin{proof}
$\Ncal$ depends on $\delta$ only through $\sqrt{\delta^{2}+4z^{2}}$, which is
even in $\delta$ and strictly increasing in $|\delta|$; $F_{\phi}=4z^{2}/s$
has no $\delta$ dependence.
\end{proof}

This delimits the geometry precisely: two coordinates $(z,s)$ suffice for the
QFI everywhere on the X-state family, whereas the negativity requires a third
coordinate $\delta$ once leakage is asymmetric. The common description of
Sec.~\ref{sec:separation} is thus complete for symmetric leakage and exactly
characterized where it stops, which is a domain statement, not a caveat
(Fig.~\ref{fig:asym}). Physically, one-sided amplitude damping is the natural
realization: it breaks $u=v$, raising the negativity above the symmetric
value at the same coherence while leaving the phase QFI untouched.

\begin{figure}[t]
\includegraphics[width=\columnwidth]{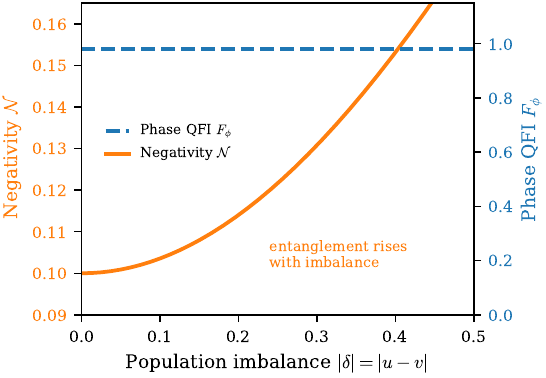}
\caption{Asymmetric leakage at fixed phase-bearing coherence $z=0.35$ and
fixed total leakage $u+v=0.5$. The phase QFI (blue dashed, right axis) is
constant, invariant under the imbalance $\delta=u-v$, while the negativity
(orange, left axis) rises monotonically with $|\delta|$
(Proposition~\ref{prop:extremal}). The QFI is a two-coordinate quantity
$(z,s)$; the negativity requires the third coordinate $\delta$, which is
invisible to phase estimation.}
\label{fig:asym}
\end{figure}

\section{Operational meaning, measurement, and direct probe}
\label{sec:operational}

\emph{Why the QFI.} A nonzero $F_{\phi}$ means the output distribution
depends on $\phi$, and the quantum Cram\'er--Rao bound
$\mathrm{Var}(\hat\phi)\geq1/(MF_{\phi})$ certifies that this dependence is
estimable: from $M$ copies the phase can be resolved to
$O(1/\sqrt{MF_{\phi}})$. Positive QFI beyond ESD therefore means a
\emph{separable} output still supports estimation of the encoded phase. We
phrase the claim narrowly: output entanglement is not necessary for retaining
nonzero phase information \emph{in this task}, and we do not generalize to all
metrology.

\emph{Attaining measurement.} The bound is attained by a projective
measurement in the eigenbasis of the symmetric logarithmic derivative, which
for the even-parity block reduces to an interference measurement between
$\ket{00}$ and $\ket{11}$; the optimal basis depends on the working point
$\phi_{0}$ through $(\ket{00}\pm ie^{-i\phi_{0}}\ket{11})/\sqrt2$, completed
by the odd-sector populations, so it is realized by local estimation around a
consistent estimate with adaptive updating (Appendix~\ref{app:block}). This
is the decoding-CNOT-plus-equatorial-readout of the companion protocol.

\emph{Direct single-qubit probe.} If the goal is only to estimate the phase
of qubit $A$, why convert? The direct probe has $F_{\rm direct}=4z_{A}^{2}$
[the $s=1$ case of Proposition~\ref{prop:block}]. Under matched noise
exposure, the converted-to-direct ratio $R_{F}=F_{\rm conv}/F_{\rm direct}$
obeys: for dephasing at equal storage time, $R_{F}=1-p<1$ (the two-qubit
coherence decays twice as fast); for a depolarizing gate error,
$R_{F}=(1-p)^{2}/(1-p/2)<1$; and for amplitude damping,
$R_{F}=(1-\gamma)/[1-2\gamma(1-\gamma)\sin^{2}\theta]$, exceeding unity only
in the window $\sin^{2}\theta>1/[2(1-\gamma)]$, a \emph{relative} robustness
that occurs where the absolute sensitivity is itself small and that ignores
the ancilla, gate, and measurement resources conversion consumes
(Fig.~\ref{fig:direct}). The implication is limited: conversion is
not a sensitivity-enhancement strategy. Its role is to provide a controlled
setting in which the separation between phase sensitivity and entanglement is
exact. The results are information-retention benchmarks, not a precision
advantage.

\begin{figure}[t]
\includegraphics[width=\columnwidth]{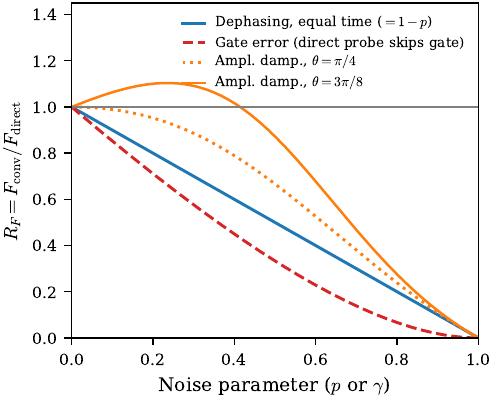}
\caption{Converted-to-direct QFI ratio under matched exposure. Dephasing
(blue) and gate error (red dashed) make conversion strictly costly;
amplitude damping exceeds unity only in the window
$\sin^{2}\theta>1/[2(1-\gamma)]$ (orange, at $\theta=\pi/4$ and $3\pi/8$), a
relative robustness that is not a metrological advantage after resource
accounting.}
\label{fig:direct}
\end{figure}

\section{During-gate benchmark}
\label{sec:during}

The analysis so far idealizes the CNOT as instantaneous and noise-free,
with noise acting only afterward. To bound this idealization we compare
against a during-gate model in which the entangling Hamiltonian and a fixed
set of Lindblad operators act simultaneously (Appendix~\ref{app:during}
specifies the generator, rates, gate duration, and the matched-exposure
prescription: equal integrated rate to the post-gate channel parameter). For
the Hamiltonian, initial state, matching prescription, and Lindblad
generators considered, the post-gate channel gives a \emph{lower} benchmark
for both the final QFI and the negativity: during part of the gate the phase
still resides in less fragile single-qubit coherence, whereas the post-gate
construction exposes the fully formed two-qubit coherence, which decays
fastest, to the entire noise budget. We claim no universal bound, only that, in the models tested, the
post-gate formulas are conservative.

\section{Discussion and limitations}
\label{sec:discussion}

\emph{Why the QFI survives ESD.} Entanglement death means the phase-bearing
coherence no longer exceeds the population penalty in the partial-transpose
spectrum; it does \emph{not} mean the coherence is zero. The QFI, blind to the
population penalty except through the mild block-weight renormalization,
remains positive as long as any coherence survives.

\emph{What the common geometry means.} Under symmetric leakage both resources
are computable from the shared coordinates $(z,g)$, but they are not the same
resource: one is a thresholded difference, the other a smooth ratio. The
shared coordinates explain why the two can be compared exactly; the different
operations explain why they fail at different points.

\emph{Operational implication.} Entanglement verification is not sufficient to
diagnose phase-information loss in this protocol. A certificate that the
output has become separable says nothing about whether $\phi$ remains
estimable, and by Theorem~\ref{thm:separation} there is an entire region
where it does.

\emph{Scope and nonclaims.} The analysis is limited to single-parameter phase
estimation on the two-qubit phase-encoded X-state family under the specified
channels, with a fixed bounded generator (so $F_{\phi}\leq1$ and no scaling is
at stake); multiparameter estimation and estimation of the noise parameters
are not treated. We claim no novelty for block-QFI evaluation, for coherence
supporting phase estimation, for nonzero QFI on separable states, or for the
negativity geometry of Ref.~\cite{Companion}; and no metrological advantage,
the direct probe being superior under matched exposure. The during-gate
conclusion is specific to the tested generators and matching prescription.

The content specific to this work is: the resource-separation theorem and the exact
separable-but-phase-sensitive region (Theorem~\ref{thm:separation}); the
death-point sensitivity \eqref{eq:Fstar} and its equivalence classification
(Corollaries~\ref{cor:equiv}--\ref{cor:onesixth}), including the structural
$\tfrac16$; the extremality of symmetric leakage for entanglement and the
resulting exact boundary of the two-coordinate description
(Proposition~\ref{prop:extremal}); and the attaining measurement with the
matched direct-probe comparison.

\section{Conclusion}
\label{sec:conclusion}

We summarize the main points. First, post-gate noise can remove the output
entanglement before it removes the encoded phase information: there is an
exact region of state space in which the output is separable yet
phase-sensitive. Second, this separation is not a coincidence of the chosen
channels but a consequence of two functional forms reading the same
coordinates: a thresholded difference for the negativity, a smooth
normalized square for the QFI. Third, the residual sensitivity at
entanglement death is a function of the death coordinate alone, so
channels dying at the same coordinate share their residual QFI; global and
independent local depolarization form one such class with $\tfrac16$ of the
noiseless sensitivity surviving at maximal input coherence. Fourth,
asymmetric population transfer adds a third coordinate that changes the
entanglement while leaving the QFI invariant, which marks the domain where
the two-coordinate description holds. Taken together, these results show that
entanglement death is not phase-information death, and quantify how much
sensitivity survives and where the two-coordinate description applies.

\begin{acknowledgments}
The authors acknowledge support from the Qatar Center for Quantum Computing
at Hamad Bin Khalifa University.
\end{acknowledgments}

\section*{Data availability}
No data were created or analyzed in this study. All results are analytic
and are reproducible from the expressions in the paper together with the
numerical-verification procedure described in Appendix~\ref{app:numerics}.
\appendix

\section{Block QFI, SLD, and attaining measurements}
\label{app:block}

\subsection{Proof of Proposition~\ref{prop:block}}

The state \eqref{eq:generalX} is block diagonal: the even block
$\rho_{2}=\bigl(\begin{smallmatrix}a & ze^{-i\phi}\\ ze^{i\phi} &
b\end{smallmatrix}\bigr)$ with eigenvalues
$\lambda_{\pm}=\tfrac12(s\pm\Delta)$,
$\Delta=\sqrt{(a-b)^{2}+4z^{2}}$, $s=a+b$, plus eigenvalues $u,v$ on
$\ket{01},\ket{10}$. For the unitarily generated family of Sec.~\ref{sec:setup} the QFI is
\begin{equation}
F_{\phi}=2\sum_{\lambda_{i}+\lambda_{j}>0}
\frac{(\lambda_{i}-\lambda_{j})^{2}}{\lambda_{i}+\lambda_{j}}
|\bra{i}G\ket{j}|^{2}.
\label{eq:qfidef}
\end{equation}
Since $G\ket{01}=G\ket{10}=0$, all terms involving the odd sector vanish,
as do diagonal terms. Within the block, write
$\rho_{2}=\tfrac12(s\openone+\bm w\cdot\bm\sigma)$ with
$w_{\perp}=2z$, $|\bm w|=\Delta$, and
$G=\tfrac12(\openone-\sigma_{z})$ there; then
$|\bra{+}G\ket{-}|^{2}=\tfrac14\sin^{2}\Theta=z^{2}/\Delta^{2}$ with
$\Theta$ the polar angle of $\bm w$, and \eqref{eq:qfidef} gives
$F_{\phi}=4(\Delta^{2}/s)(z^{2}/\Delta^{2})=4z^{2}/s$, independent of
$u,v$ individually. Rank-deficient cases follow by continuity with the
excluded-pair convention. For a pure even-block state this reduces to
$4\,\mathrm{Var}(G)$, consistent with the noiseless value $F_{0}=\sin^{2}2\theta$.

\subsection{Symmetric logarithmic derivative}

$L$ solves $\partial_{\phi}\rho=-i[G,\rho]=\tfrac12(L\rho+\rho L)$; one
may take $L=0$ on the odd sector and, in the block eigenbasis,
$L_{+-}=2(\lambda_{+}-\lambda_{-})\bra{+}(-iG)\ket{-}/(\lambda_{+}
+\lambda_{-})$, whence $F_{\phi}=\mathrm{Tr}(\rho L^{2})=4z^{2}/s$. The
eigenbasis of $L$ defines the optimal projective measurement; at the
working point it is the equatorial block basis used below, and it is
channel independent because the block coherence phase is channel
independent.

\subsection{Attaining measurements}

\emph{Parity-resolving projection.} The four-outcome measurement
$\Pi_{\pm}=\ket{\pm_{y}}\bra{\pm_{y}}$ with
$\ket{\pm_{y}}=(\ket{00}\pm ie^{-i\phi_{0}}\ket{11})/\sqrt2$, plus
$\Pi_{01},\Pi_{10}$, has outcome probabilities
$p_{\pm}=s/2\pm z\sin(\phi-\phi_{0})$ and $\phi$-independent $u,v$; at
$\phi=\phi_{0}$ the classical Fisher information is
$2z^{2}/(s/2)=4z^{2}/s=F_{\phi}$.

\emph{Decoding circuit.} The CNOT maps
$\ket{00}\!\to\!\ket{00}$, $\ket{11}\!\to\!\ket{10}$,
$\ket{01}\!\to\!\ket{01}$, $\ket{10}\!\to\!\ket{11}$: the even block
becomes a state of qubit $A$ with $B$ heralding $0$ (probability $s$), the
odd sector heralds $1$. Conditioned on the herald, qubit $A$ carries
coherence $z/s$; the equatorial single-qubit measurement then yields
classical Fisher information $4(z/s)^{2}$ per heralded copy, and
$s\cdot4(z/s)^{2}=4z^{2}/s$ per copy overall. Saturation is local in
$\phi$; adaptive updating of $\phi_{0}$ applies as usual.

\section{Residual formulas}
\label{app:residual}

For the depolarizing class, $g=(1-f)/4$ and $g_{\star}=f_{\star}\kappa$
give $f_{\star}=1/(4\kappa+1)$, $g_{\star}=\kappa/(4\kappa+1)$,
$s_{\star}=(2\kappa+1)/(4\kappa+1)$, hence
\begin{equation}
F^{\star}=\frac{4(f_{\star}\kappa)^{2}}{s_{\star}}
=\frac{4\kappa^{2}}{(4\kappa+1)(2\kappa+1)}
=\frac{s_{2}^{2}}{(2s_{2}+1)(s_{2}+1)},
\end{equation}
with $s_{2}=\sin2\theta=2\kappa$; at $\theta=\pi/4$,
$F^{\star}=\eta^{\star}=\tfrac16$. The thresholds
$p_{c}=2s_{2}/(2s_{2}+1)$ and $p_{c}^{\rm loc}=1-1/\sqrt{2s_{2}+1}$
\cite{Companion} differ; the residual is a property of the common point.
For amplitude damping at $\gamma_{c}=\cot\theta$ ($\theta>\pi/4$):
$f\kappa=\cos\theta(\sin\theta-\cos\theta)
=\tfrac12(s_{2}-1-\cos2\theta)$ and $s=1-2f\kappa$, giving
the residual \eqref{eq:Fstar}, which vanishes at both endpoints
$\theta\to\pi/4^{+}$ and $\theta\to\pi/2$ and is maximal between them
[Fig.~\ref{fig:death}].

\section{Fingerprint slopes at general input angle}
\label{app:slopes}

From $\eta_{F}=f^{2}/(1-2g)$ and $\eta_{\Ncal}=f-g/\kappa$ (before
death): $\eta_{F}'(0)=2f'(0)+2g'(0)$ and
$\eta_{\Ncal}'(0)=f'(0)-g'(0)/\kappa$. Table~\ref{tab:fg} gives
\begin{center}
\begin{tabular}{lcc}
\hline\hline
Channel & $\eta_{\Ncal}'(0)$ & $\eta_{F}'(0)$\\
\hline
Phase damping & $-1$ & $-2$\\
Global depolarizing & $-1-\dfrac{1}{2\sin2\theta}$ & $-\dfrac32$\\[6pt]
Amplitude damping & $-(1+\tan\theta)$ & $-2\cos^{2}\theta$\\[2pt]
Local depolarizing & $-2-\dfrac{1}{\sin2\theta}$ & $-3$\\
\hline\hline
\end{tabular}
\end{center}
The two-sector degeneracy of amplitude damping and global depolarizing
requires $-(1+\tan\theta)=-1-1/(2\sin2\theta)$, i.e.,
$4\sin^{2}\theta=1$, and $-2\cos^{2}\theta=-\tfrac32$, i.e.,
$\cos^{2}\theta=\tfrac34$: both give $\theta=\pi/6$, with common pair
$(-1-1/\sqrt3,\,-\tfrac32)$. At $\theta=\pi/4$ ($\sin2\theta=1$) the table
gives the slope pairs $(\eta_{\mathcal N}',\eta_{F}')=(-1,-2)$,
$(-\tfrac32,-\tfrac32)$, $(-2,-1)$, and $(-3,-3)$ for the four channels.

\section{Channels, asymmetric entries, and physical mappings}
\label{app:channels}

\emph{Symmetric channels.} Phase damping multiplies the $00$--$11$
coherence by $1-p$; global depolarizing is
$\rho\to(1-p)\rho+p\,\openone/4$; amplitude damping applies
$K_{0}=\mathrm{diag}(1,\sqrt{1-\gamma})$,
$K_{1}=\sqrt{\gamma}\ket{0}\bra{1}$ on each qubit; local depolarizing
applies the single-qubit depolarizing channel with parameter $p$ on each
qubit. The resulting $(f,g)$ pairs are those of Table~\ref{tab:fg}
\cite{Companion}.

\emph{Asymmetric amplitude damping.} With strengths
$\gamma_{A},\gamma_{B}$, the input
$\cos\theta\ket{00}+e^{i\phi}\sin\theta\ket{11}$ maps to the X state with
entries of Sec.~\ref{sec:asym} and
$b=(1-\gamma_{A})(1-\gamma_{B})\sin^{2}\theta$,
$a=1-b-u-v$; no odd-sector coherence is generated, and all entries were
verified against the Kraus map to machine precision.

\emph{Generalized amplitude damping.} With Kraus operators
$\sqrt{1-\bar N}\{K_{0},K_{1}\}$ and
$\sqrt{\bar N}\{K_{0}^{T},K_{1}^{T}\}$ per qubit, transitions occur with
probabilities $\gamma(1-\bar N)$ ($1\to0$) and $\gamma\bar N$
($0\to1$); the coherence factor per qubit remains $\sqrt{1-\gamma}$,
independent of $\bar N$. Applied to both qubits the output remains a
symmetric X state with $f=1-\gamma$ and $g$ the generalized-damping penalty.

\emph{Physical mappings.} $p=1-e^{-2t/T_{\varphi}}$ (each qubit dephasing
at rate $1/T_{\varphi}$), $\gamma=1-e^{-t/T_{1}}$, and, for a gate of
average fidelity $F_{\rm avg}$ \cite{Magesan2011} modeled by global depolarization,
$p=\tfrac43(1-F_{\rm avg})$; a gate error
$\varepsilon_{\rm gate}=0.5\%$ gives $p\approx6.7\times10^{-3}$ and a QFI
reduction $\approx\tfrac32p\approx1\%$. See the caveat of
Sec.~\ref{sec:operational}.

\section{During-gate model}
\label{app:during}

The during-gate dynamics is
$\dot\rho=-i[H,\rho]+\sum_{j}\Gamma_{j}\mathcal{D}[L_{j}]\rho$ with
$H=[\pi/(2t_{g})]\ket{1}\bra{1}_{A}\otimes(\openone-X_{B})$, which
generates the CNOT at $t=t_{g}$, and
$\mathcal{D}[L]\rho=L\rho L^{\dagger}-\tfrac12\{L^{\dagger}L,\rho\}$. For
dephasing, $L=\sqrt{1/(2T_{\varphi})}\,\sigma_{z}$ on each qubit (so a
single-qubit coherence decays as $e^{-t/T_{\varphi}}$); for relaxation,
$L=\sqrt{1/T_{1}}\,\sigma_{-}$ on each qubit. The comparison partner is
the ideal CNOT followed by the corresponding channel of
Table~\ref{tab:fg} with $p=1-e^{-2t_{g}/T_{\varphi}}$ or
$\gamma=1-e^{-t_{g}/T_{1}}$. The master equation is integrated by
fourth-order Runge--Kutta; the QFI of the output family is computed from
$\rho(\phi)$ and $\partial_{\phi}\rho$ via the spectral representation of
the SLD, and the negativity from the partial transpose. Representative
values at $\theta=\pi/4$: for $t_{g}/T=0.2$, dephasing gives
$F=0.570$ (during) versus $0.449$ (post) and $\Ncal=0.360$ versus
$0.335$; relaxation gives $F=0.803$ versus $0.787$ and $\Ncal=0.397$
versus $0.335$; the largest non-X element and the imbalance $|\delta|$
remain below $6\times10^{-2}$ for $t_{g}/T\leq0.6$.

\section{Numerical verification}
\label{app:numerics}

All analytical expressions of the paper, namely the general formulas
\eqref{eq:blockqfi} and \eqref{eq:generalN} over randomly sampled states
including $u\neq v$, the closed forms of Corollary~\ref{cor:closed} on a
$121\times101$ grid per channel, the asymmetric-leakage relations of Sec.~\ref{sec:asym},
the generalized-damping coordinates including the exact
symmetry $u=v$ and preservation of the X form, the residuals
\eqref{eq:Fstar}, the channel closed forms, and the
single-qubit benchmark $F_{\rm direct}=4z_A^2$, were independently checked
against direct density-matrix calculations, with discrepancies at machine
precision throughout; the verification scripts accompany the manuscript.

\end{document}